\documentclass[sigconf,10pt,nonacm,acmthm=false]{acmart}

\setcopyright{none}
\renewcommand\footnotetextcopyrightpermission[1]{}

\usepackage{amsmath}
\usepackage{amsfonts}
\usepackage{graphicx}

\usepackage{algorithm2e}
\DontPrintSemicolon
\RestyleAlgo{ruled}

\SetCommentSty{mycommfont}

\usepackage{todonotes}
\usepackage{subcaption}

\usepackage{amsthm}

\newtheorem{observation}{Observation}
\newtheorem{lemma}{Lemma}
\newtheorem{theorem}{Theorem}
\newtheorem{corollary}{Corollary}

\begin{document}

\title{Privacy Attacks on Stable Marriage}


\author{Stephan A. Fahrenkrog-Petersen}
\affiliation{%
  \institution{University of Liechtenstein}
  \city{Vaduz}
  \country{Liechtenstein}
}
\email{stephan.fahrenkrog@uni.li}

\author{Aleksander Figiel}
\affiliation{%
  \institution{TU Berlin}
  \city{Berlin}
  \country{Germany}
}
\email{a.figiel@tu-berlin.de}

\author{Darya Melnyk}
\affiliation{%
  \institution{TU Berlin}
  \city{Berlin}
  \country{Germany}
}
\email{melnyk@tu-berlin.de}

\author{Tijana Milentijevi\'c}
\affiliation{%
  \institution{TU Berlin}
  \city{Berlin}
  \country{Germany}
}
\email{tijana.milentijevic@tu-berlin.de}

\author{Stefan Schmid}
\affiliation{%
  \institution{TU Berlin and Weizenbaum Institute}
  \city{Berlin}
  \country{Germany}
}
\email{stefan.schmid@tu-berlin.de}

\begin{abstract}
The stable marriage problem appears in many privacy-sensitive domains, for example in the National Resident Matching Program in the US. In such applications, preserving the privacy of users’ preference lists is essential to prevent strategic manipulation, discourage misreporting, and comply with data protection regulations.

In this work, we investigate privacy attacks on stable marriage algorithms. Assuming that the attacker (e.g., the hospitals) can repeatedly interact with the stable marriage algorithm, we demonstrate how such interactions can reveal private preferences of the non-malicious side (e.g., the residents). We show that the widely applied Gale-Shapley Matching Algorithm, where the proposers' side is malicious, is vulnerable to privacy attacks and all honest agents' preferences can be revealed. 
We further investigate which preference distributions of the honest, non-malicious side are susceptible to privacy attacks and show that the Gale-Shapley Matching Algorithm where the honest side proposes can preserve privacy in non-susceptible preference distributions. We extend our results to the decentralized setting and show that the attacker's side can infer all preference orderings. In an experimental evaluation, we test privacy attacks on synthetic and real-world data and show that real-world data is indeed susceptible to privacy attacks. This work underlines a need for new privacy-preserving stable marriage algorithms.
\end{abstract}

\maketitle

\begin{center}
\begin{minipage}{0.95\columnwidth}
\footnotesize
\copyright~2026 IEEE. Personal use of this material is permitted.
Permission from IEEE must be obtained for all other uses, in any
current or future media, including reprinting/republishing this
material for advertising or promotional purposes, creating new
collective works, for resale or redistribution to servers or lists,
or reuse of any copyrighted component of this work in other works.
\end{minipage}
\end{center}

\begin{center}
\begin{minipage}{0.95\columnwidth}
\footnotesize
This research was supported by the German Research Foundation
(DFG), Priority Programme SPP 2378 (ReNO-2), 2025--2029.
\end{minipage}
\end{center}

\section{Introduction}
 
In many real-world applications, such as college admissions, donor-patient pairing, job markets, or residency assignment, it is required to compute a \emph{stable marriage}, also referred to as the stable matching, that reflects the preferences of the participants. These preferences are typically expressed as ranked lists that contain sensitive information and should be kept private.
A simple everyday example is assigning students to project topics: students rank the projects they would like to work on, while supervisors rank the students they would like to supervise. A matching is stable if there is no student and supervisor who would both rather work with each other than with their assigned partner. 
In order to compute a stable marriage, algorithms like the renowned Gale-Shapley algorithm let participants on one side propose to their top-ranked preference. If rejected, they iteratively propose to the next best alternative in their ranking.

A particularly relevant and sensitive use case is the US National Resident Matching Program (NRMP), where medical graduates are matched to hospital residency programs. In this setting, applicants (graduates) are the more vulnerable side, while hospitals may have incentives to prefer candidates who rank them highly. For example, a clinic may avoid selecting applicants who list it as a second or third choice, fearing those candidates will leave after training. If such programs could infer or access parts of the applicants' preferences, they might manipulate the process to their advantage. This undermines both fairness and stability, and highlights the need for matching protocols that maintain not only correctness, but also the confidentiality of individual preferences.

\paragraph{Our Contribution} We initiate the study of how much information one colluding side of the participants (the hospitals in the example) can infer about the preferences of the other side of the participants (the residents in the example). 
Note that we assume a strong collaboration on the colluding side, with full information exchange.
The challenge lies in choosing preferences for the dishonest (colluding) side in such a way that useful information about the preferences of the honest side can be extracted from each stable matching. 
Given a single stable matching, an attacker may infer some preference information such as one participant ranking another above someone else.
However, in practice, many matching systems repeat the matching process over multiple rounds, allowing the change of one's preferences between the rounds. We show that this poses a great risk to the confidentiality of preference information. 
We analyze both the centralized and the decentralized settings. 
In the centralized case, a trusted third party aggregates all preferences and publicly announces a stable matching. In the decentralized setting, participants send each other matching proposals in a peer-to-peer fashion and thus determine their matching partner.
In this work, we assume that each side has $n$ participants.
We show that:
\begin{itemize}
    \item it is possible to infer the top preference of all honest participants within $n$ matching rounds, independent of the chosen algorithm,
    \item it is possible to learn the full preference ordering of participants if all participants have the identical preferences,
    \item there are algorithms that, under certain input preferences can prevent the attacker from revealing too much information, such as the Gale-Shapley matching algorithm where the honest side proposes,
    \item the Gale-Shapley matching algorithm where the dishonest side proposes, as well as the decentralized Gale-Shapley matching algorithms, are vulnerable to revealing the full preference ordering.
\end{itemize}
We further design multiple strategies for the centralized and decentralized setting that aim to maximize what can be learned about the preferences and evaluate them empirically on synthetic and real-world data.
 

\paragraph{Example}
Consider the problem of pairing medical graduates with residency programs. In the US, the pairing is done using the National Resident Matching Program (NRMP), also known as The Match. The algorithm used in the program is based on the traditional Gale-Shapley algorithm where applicants propose, delivering an applicant-optimal solution. Both applicants and clinics submit their preferences to the R3 system, which performs the pairing centrally. Now hospitals might prefer candidates who rank them highly, fearing that applicants who ranked them lower would leave after training. Therefore, residency programs may have an incentive to learn applicants' preferences and exploit this information to their advantage. As we demonstrate, by using different preference orderings and repeating the matching, the clinics are able to find out applicants' preferences. Once the preferences are learned, the clinics can adjust their own preferences strategically to get a match in their favor. For a toy example on how the learning process for an attacker works see Figure~\ref{fig:learning_example}.

\begin{figure}
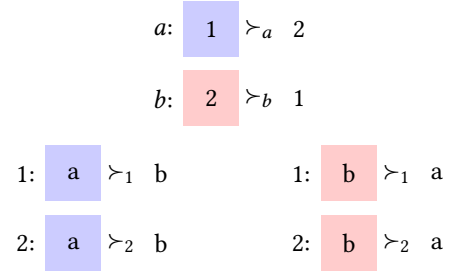

    \centering
    \begin{subfigure}{0.2\textwidth}
        \begin{align*}
        \makebox(15,15){$a$:} \colorbox{blue!20}{\makebox(15,15)1} \makebox(15,15){$\succ_a$} \makebox(15,15)2 \\
        \makebox(15,15){$b$:} \colorbox{red!20}{\makebox(15,15)2} \makebox(15,15){$\succ_b$} \makebox(15,15)1
        \end{align*}        
    \end{subfigure}
    \\
    \begin{subfigure}{0.2\textwidth}
        \begin{align*}
        \makebox(15,15){1:} \colorbox{blue!20}{\makebox(15,15)a} \makebox(15,15){$\succ_1$} \makebox(15,15)b \\
        \makebox(15,15){2:} \colorbox{blue!20}{\makebox(15,15)a} \makebox(15,15){$\succ_2$} \makebox(15,15)b
        \end{align*}
    \end{subfigure}
    \begin{subfigure}{0.2\textwidth}
        \begin{align*}
        \makebox(15,15){1:} \colorbox{red!20}{\makebox(15,15)b} \makebox(15,15){$\succ_1$} \makebox(15,15)a \\
        \makebox(15,15){2:} \colorbox{red!20}{\makebox(15,15)b} \makebox(15,15){$\succ_2$} \makebox(15,15)a
        \end{align*}
    \end{subfigure}
    \caption{A small example of how one can attack the preferences of agents $a$ and $b$. The agents 1 and 2 are ranked highest by $a$ and $b$, respectively (see top figure). If both attackers 1 and 2 lie about their preferences and set $a$ as their top preference, as visualized in the bottom left figure, then, in any stable matching, $a$ has to be matched to its first preference. Once we observe the matching for these preferences, and see that $a$ and $1$ are matched, we conclude $1$ was the first preference of $a$. Likewise, the attackers can both set $b$ as their first preference, as visualized in the bottom right figure, and learn that $2$ is the first preference of $b$. Since this example is so small, once we know the first preference, we also know the last. In two matching iterations, the attackers were able to learn all preferences.}
    \label{fig:learning_example}
\end{figure}

\section{Related Work} 
The stable marriage problem was formally introduced by Gale and Shapley in \cite{gale1962college} and is today often referred to as the stable matching problem. Later variants of this problem, such as allowing partial preferences or having ties in the preference list, have been considered in \cite{stable-marriage-book}. 

Adversarial behavior in matching problems in form of strategic manipulation and misreporting of preferences has been studied extensively in the literature \cite{roth1982economics, stable-marriage-book}.
\cite{roth1982economics} showed that the stable marriage is not always truthful, as agents could lie about their preferences in order to get a better matching partner. Nevertheless, \cite{gale1962college} have shown the proposing side cannot gain anything by misreporting their preferences in the Gale-Shapley algorithm. 
A similar setting was considered in \cite{teo2001gale}, where it is required that all agents submit complete (not partial) preference lists. They showed that an optimal manipulation strategy for women, under the man-optimal Gale–Shapley outcome, can be computed in polynomial time. Additionally, cases where agents from the proposing side can work jointly to obtain better partners \cite{huang2006cheating} and construct stable matching pairs even with errors in the input \cite{mai_et_al:LIPIcs.ESA.2018.60} were studied in the past. 
Also robustness measures, with respect to small changes in preferences,  related to stable matchings were examined in \cite{boehmer2024worst, 10.1145/3485000, drummond2013elicitation, aziz2017stable, aziz2020stable}. The aforementioned lying strategies assume the knowledge of preference orderings based on which a better pair for an agent can be obtained.
\cite{constantinescu2025byzantinestablematching} recently introduced the notion of Byzantine stable matching. The authors study the solvability of the stable matching problem in the presence of malicious parties on both sides, examining how different communication models impose varying conditions for solvability. 
In contrast to the previous adversarial studies, our work focuses on inferring preference lists based on the pairs matched by an algorithm rather than manipulating the matching.


Privacy aspects of stable marriage have also been considered in the literature. A stable matching algorithm is considered private if it computes a stable match while revealing no additional information to the adversary, beyond what can be inferred from the final matching itself and the preferences of the participants under the adversary’s control \cite{10.1007/11889663_5}. \cite{10.1007/11889663_5} proposed an algorithm where privacy is guaranteed if the majority of matching authorities are honest. 
Further work regarding private matchings using secret sharing schemes \cite{franklin2007improved} and homomorphic encryption \cite{teruya2015round} have been studied. Practically implemented private stable matchings with thousands of agents include \cite{10.1145/2976749.2978373, teruya2015round, keller2014efficient, blanton2013data, zahur2016revisiting, riazi2017toward}.

Research on multi-agent systems provides an additional perspective on private matching problems, particularly through work on strong collaboration among agents \cite{wooldridge2009introduction, sioutis2006agent, 10.1145/3745238.3745531, lashkari1997collaborative}, where agents can collaborate to reach their own goals. 
In this context, matchmaking in multi-agent systems refers to mechanisms that assign agents to other agents or clusters based on their interests and preferences \cite{foner1997yenta, 10.1145/545056.545129, kuokka1995matchmaking}. Unlike our work, these scenarios are typically benign and assume all agents are honest and behave correctly. Note that in our work, benign agents do not collaborate, but we do assume strong collaboration among dishonest agents.

The privacy attack considered in this paper was first introduced in the context of scheduling problems \cite{fahrenkrog2023privacy}, where the original scheduling instance is solved and the privacy loss resulting from the publication of the final schedule is analyzed. 
To our best knowledge, we are the first to consider privacy attacks in stable marriage where only the final matching is revealed and the protocol itself does not expose intermediate state. 


\section{Privacy Attacks} 

\subsection{Model}
\paragraph{Stable marriage/matching}
We consider a two-sided matching system with $2n$ participating agents, where every agent has to be matched to exactly one agent from the other side. The agents are divided into two groups, $H = \{ h_1, \dots, h_n\}$ representing honest agents, and $D = \{d_1, \dots, d_n\}$ denoting dishonest agents, such that $|H|=|D|=n$.
We assume that the dishonest agents share all information with each other and coordinate on how to set their preferences. For simplicity one may assume a single entity controls the dishonest agents.
We assume that each agent in $a\in H$ (resp. $D$)  maintains a strict, linear preference ordering $\pi_a$ over the agents in $D$ (resp. $H$). In particular, for any pair of agents $b,c \in D$, $a\in H$ either strictly prefers agent $b$ over $c$, denoted $b \succ_a c$, or it strictly prefers agent $c$ over $b$, denoted $c \succ_a b$. We represent the ordering of agent $a$ as a list  $\pi_a = [1^a, 2^a, \dots, n^a]$, where $1^a$ is the agent that $a$ ranks the highest, and $n^a$ is the agent $a$ ranks the lowest.
We use $\Pi^H$ (resp. $\Pi^D$) to denote the matrix containing the preferences of all agents $a \in H$ (resp. $D$), specifically $\Pi^D[d,i] = h$ is the honest agent $h \in H$ that is $i^{\text{th}}$ in the ranking of $d \in D$
The goal is to find a \textit{stable matching} $M$ between $H$ and $D$: a stable matching is a matching that does not contain any \textit{blocking} matching pairs. A pair of agents $(h,d)\in H\times D$ is called blocking, if both $h$ and $d$ prefer to be matched to each other over their current matching partner. In the following we define two models in which we consider the stable matching problem.

\paragraph{Centralized model} In the centralized setting, we assume that there is a trusted central server to which the agents submit their preference rankings. Upon receiving the preferences $\Pi^H$ and $\Pi^D$, a matching algorithm is executed on the server to compute a stable matching $M$. Afterward, the agents receive their matching partners from the server. In this work, we will assume that the applied algorithm to compute a stable matching is public, and thus known to all agents. 

\paragraph{Decentralized model} In the decentralized setting, we assume that the agents can communicate to each other in a peer-to-peer manner. The communication is assumed to be synchronous, i.e. all agents start executing a stable matching algorithms at the same time and proceed in rounds. We assume that all agents know the executed algorithm.  

\paragraph{Adversarial model} In this work, we differentiate between two types of agents. The set $H$ denotes a set of \textit{honest} agents who always share their true preference lists with the central server in the centralized case. In the decentralized case, we assume that the honest agents do not change their preference lists during different algorithm runs, and that they always correctly follow the algorithm. The set $D$ denotes the set of \textit{dishonest} (or \textit{malicious}) agents. These agents are not bound to sharing their true preferences with the central server or with other agents. Additionally, in the decentralized case, malicious agents may not follow the distributed algorithm correctly. 

\paragraph{Centralized privacy attack} We assume that the goal of an attacker (who represents the malicious agents) is to violate the privacy of the honest agents. This means that the goal of the dishonest agents is to (at least partially) reveal the preference lists of the honest agents. In the centralized case, this can be done by repeatedly executing a matching algorithm with the identical preference matrix $\Pi^H$, and adaptively changing dishonest preferences $\Pi^D$. To compare different malicious attacks, we assume that one run of the matching algorithm by the server is atomic. An attack that finds out the preferences of the honest agents with fewer matching computations is thus favorable.
 
\paragraph{Decentralized privacy attack} While the goal of a malicious adversary is the same as in the centralized case, we change the way to estimate the efficiency of an attack. In the decentralized case, we count the number of rounds needed to reveal the preferences of the honest agents. If multiple algorithm executions are needed, we sum the number of rounds used in each execution.

\section{Theoretical Results}

\subsection{Gale-Shapley Matching Algorithms}

In this section, we will describe different variants of the classical Propose$\&$Reject \cite{gale1962college} algorithm that will be analyzed in the following sections.

\paragraph{Propose$\&$Reject - centralized}
The algorithm is executed by a central authority that collects complete preference lists from all agents and outputs a stable matching. In particular, one side (noted as the proposers) iteratively proposes to agents on the other side according to the preference ordering, while each receiver tentatively accepts the most preferred proposal so far and rejects the rest. The process continues until no further proposals are made and the central authority outputs a stable matching.

\paragraph{Propose$\&$Reject - decentralized}
This algorithm is a generalization of the centralized algorithm, and it is executed in synchronous rounds. In the first round, the proposing side proposes to agents on the other side according to the preference ordering. Next, each receiver tentatively accepts the most preferred proposal so far and rejects the rest. The process proceeds in rounds and terminates when no new proposals are made, resulting in a stable matching.


\subsection{Algorithm-Independent Leakage}
In this section, we show privacy attacks that generalize to all matching algorithms. We start by proving that, independent of the choice of the preference matrix $\Pi^H$, there exists an adversarial strategy that reveals the top preference of each honest agent.

    \begin{lemma} \label{lem:first-choice-learnable}
    Consider the centralized stable matching setting. Independent of the stable matching algorithm, it is always possible for the malicious side to learn the top preference of each agent in $H$ within $n$ iterations.   
    \end{lemma}
    \begin{proof}
    Assume an adversarial strategy where, in the first iteration, every agent $d \in D$ sets the first agent $h_1$ of $H$ as the top choice for every candidate. This way, any stable matching algorithm has to match $h_1$ with its top choice. If this was not the case, $h_1$ and its top preference would form a blocking pair. Similarly in iteration $i$, the malicious agents place the $i^{\text{th}}$ agent in $H$ as their top alternative. After $n$ rounds the adversary knows the top preferences in of every $h \in H$.
    \end{proof}
Next, we consider the case where all good agents have the identical preferences in  $\Pi^H$.

    \begin{lemma} \label{lem:all-same-good-propose}
    If every row of $\Pi^H$ is the same, i.e., all good agents have the identical preferences, then there exists no stable matching algorithm that can hide $\Pi^H$.
    \end{lemma}
    \begin{proof}
    By~\cite{stable-marriage-book, clark2006uniqueness}, there exists a unique stable matching if agents of at least one side all have the identical preference list. Since the matching is unique, we choose a stable matching algorithm and prove that this algorithm cannot hide $\Pi^H$. Consider therefore the Propose$\&$Reject algorithm where the attacker proposes. The attacker can learn the whole preference matrix in $n$ iterations using the following strategy: 
    Assume that there are $n$ honest agents with IDs $h_1, h_2,\dots, h_n$ and all of them have the identical preference orderings. 
    Let the $n$ malicious agents set their preferences in a \textit{Round-Robin} manner for the next $n$ iterations, where each agent has the identical preference list. For example, they set $h_1 \succ h_2 \succ h_3 \succ  \dots  \succ h_n$ in the first iteration; $h_2 \succ h_3 \succ  \dots h_n \succ h_1$ in the second iteration; $h_n \succ h_1 \succ  \dots  \succ h_{n-1}$ in the last iteration. We refer to the corresponding attack as the \textit{Round-Robin attack}.
    
    With this strategy, in each iteration, every honest agent will be paired with a different malicious agent that it has not been paired with in an earlier iteration. 
    This happens, since all honest agents have the identical preference list and, in each execution of a matching algorithm, they will be proposed to by all malicious agents. For example, in iteration $3$, malicious agents have the preference list  $h_3 \succ h_4 \succ \dots  \succ h_n \succ h_1 \succ h_2$. In this case, the agent $h_3$ will be proposed to first by all malicious agents and will pair with its highest preference. Similarly, agent $h_4$ would pair with its second hightest preference and so on. Note that the matching is different in every execution of a matching algorithm, because every honest agent has the identical preference list.

Consider iteration $k$ of the presented attack strategy. Assume agents $h_i$ and $d_j$ are matched. We know that $d_j$ ranks $h_i$ at position $x = (i+k-1) \mod n + 1$ in this iteration. From the stability of the matching we can conclude that $d_j$ is in $\{1^{h_i}, \ldots, x^{h_i}\}$.
    Now after $n$ iterations of the attack strategy, for every honest agent $h_i$, we have obtained the following information: for every $y \in [n]$, $h_i$ was matched with some $d_{j_y}$ in some iteration and it was determined that $d_{j_y}$ is in $\{1^{h_i}, \ldots, y^{h_i}\}$. However, since $h_i$ was matched to a different agent in each iteration we deduce that $d_{j_y}$ in fact corresponds to $y^{h_i}$ of $h_i$. Consequently, this attack strategy learns all preferences. Note that the attack strategy does not assume beforehand that all honest agents' preferences are identical, however, if they are, then it succeeds in learning them.
    \end{proof}


\begin{observation}
    The Round-Robin attack also learns the top preference of each agent.
\end{observation}

\subsection{Analysis of centralized algorithms}

In this section, we analyze the privacy properties of the centralized stable marriage algorithms. We start by showing that there are special cases of $\Pi^H$, where an algorithm can prevent an adversary from revealing too many preferences.

    \begin{lemma} \label{lem:first-choice-different}
    Assume that the top choice for each agent in $H$ is different. Then, there exists a stable matching algorithm, such that no adversarial strategy can learn anything beyond the first preference. 
    \end{lemma}
    \begin{proof}
     If the top choices of agents in $H$ are different, a matching $M$ that maps the agents in $H$ with their top choices is always stable. Independent of the preferences in $D$ and the bottom $n-1$ preferences of agents in $H$, an algorithm can always output the matching $M$. This way, the malicious side cannot learn any more preferences beyond the top choices.
    \end{proof}

    Note that in Lemma~\ref{lem:first-choice-learnable} we showed that the top choices can always be learned independent of the algorithm. We now show that the above matrix $\Pi^H$ can be generalized: 
    
    \begin{corollary}\label{cor:reveal_clusters}
        Partition $D$ into disjoint sets $D_1, ..., D_k$, further partition $H$ into disjoint sets $H_1, ..., H_k$, such that $|D_i| = |H_i|$. Define $\Pi^H$ such that all agents in $H_i$ have $D_i$ as their top preferences (the ordering of the preferences in $D_i$ can vary among the agents in $H_i$). 
        Then, there exists a stable matching algorithm, such that no adversarial strategy can learn anything beyond the clusters, that is learn how agents in $H_i$ order $D\setminus D_i$.
    \end{corollary}

Next, we show that the Propose$\&$Reject algorithm, where the dishonest side proposes, is vulnerable to privacy attacks.

    \begin{theorem}\label{thm:dishonest-propose}
        Assume that the Propose$\&$Reject algorithm, where the dishonest side proposes, is used as the matching algorithm. Then, there exists a privacy attack that learns everything in $n^2$ matching iterations.
    \end{theorem}
    \begin{proof}
         By Lemma~\ref{lem:first-choice-learnable}, it is possible to learn the first preference of each agent in $n$ rounds. We use induction to show that the attacker can learn all preferences in $\Pi^H$. The top preferences that can be learned with Lemma~\ref{lem:first-choice-learnable} form the induction base. The induction step is performed for each honest agent $h$ separately. We assume that the attacker has learned the first $i, i\ge 1$ preferences of $h$. To learn the preference $(i+1)^h$, the attacker can use the following preference matrix: put $h$ in the last place of its first $i$ preferences $\{1^h,\ldots, (i)^h\}$, and place $h$ as the top choice for all other malicious agents.
         
         
         To prove that $(h, (i+1)^h)$ will be matched, we need to show that no agent in $\{1^h,\ldots, (i)^h\}$ will propose to $h$ in the matching algorithm. Observe first that $(i+1)^h$ will form a matching pair with $h$ in the first round, since it is the highest ranked agent to propose to $h$ in this round. After this round, $h$ can only reject $(i+1)^h$ if it receives a proposal from an agent in $\{1^h,\ldots, (i)^h\}$. Moreover, none of the agents in $D\setminus\{1^h,\ldots, (i)^h\}$ will propose to $h$ anymore. 
         
         We will now show that, in the following rounds, all other agents will form a stable matching among themselves and therefore no agent in $\{1^h,\ldots, (i)^h\}$ will propose to $h$.
         Assume by means of contraposition that an agent $d\in \{1^h,\ldots, (i)^h\}$ proposes to their last preference $h$. This means that the previous $n-1$ proposals of this agent have been rejected. A proposal can only be rejected by an agent $h_i\neq h$ if it accepts a proposal from another agent. That is, $h_i$ is matched to one of the remaining $n-2$ agents in $D$ (all agents excluding $(i+1)^h$ and $d$). Observe that there are only $n-2$ agents in $D$ that can make agents $H\setminus h$ reject proposals (or previously established matchings) from $d$. Thus, there must exist at least one agent in $H\setminus h$ that accepts a proposal from $d$ without rejecting it later. This is a contradiction to our previous assumption. This shows that $h$ will not receive any proposals after the first matching round. Therefore $(h, (i+1)^h)$ must form a pair in the output stable matching and thus be revealed to the attacker.
    \end{proof}

We now turn our focus to centralized Propose$\&$Reject algorithm where the honest side proposes and show that this algorithm can in fact preserve privacy in the previously discussed cases.
    \begin{theorem}\label{thm:propose-reject-privacy-preserving}
        Assume that the Propose$\&$Reject algorithm where the honest side proposes is used as the matching algorithm. This algorithm is privacy preserving with respect to Lemma~\ref{lem:first-choice-different} and Corollary~\ref{cor:reveal_clusters}. 
    \end{theorem}
    \begin{proof}
        Consider the matrix in Lemma~\ref{lem:first-choice-different} first. In the first round of the Propose$\&$Reject algorithm, all honest agents propose to their first preference in the list. By assumption, the first preference is different for every honest agent. Thus, every malicious node will receive exactly one proposal which it will accept in this round. Since every honest node has a matching partner, the algorithm terminates with a matching after one round. 

        Next, consider $\Pi^H$ from Corollary~\ref{cor:reveal_clusters}. Observe that the agents in $H_i$ will propose to agents in $D_i$ first, and the Propose$\&$ Reject algorithm will establish a stable matching between these subsets. Thereafter, the agents in $H_i$ will not propose anymore. This holds for every $i\in[k]$. The matching is stable, because none of the agents in $H_i$ will prefer any of the agents in $D\setminus D_i$.
    \end{proof}

\subsection{Analysis of decentralized algorithms}

In this section, we show that the decentralized versions of both Propose$\&$Reject algorithms are vulnerable to privacy attacks. For the analysis in the decentralized setting, we consider communication rounds within one iteration of the Propose$\&$Reject algorithm used for sending proposals and rejecting them.

\begin{lemma}\label{lem:decentralized-learn-everything-simple}
    There is a malicious privacy attack that learns all preferences in the  Propose$\&$Reject algorithm, where the honest side initiates proposals. This attack takes 1 algorithm iteration with $2n$ communication rounds.
\end{lemma}
\begin{proof}
    If the honest agents are the proposers, malicious agents can strategically reject all incoming proposals. By observing the sequence of proposals they receive over $n$ proposal rounds, the attacker can fully reconstruct the preference orderings of all honest agents. Moreover, the attacker can prevent honest agents from obtaining a matching.
\end{proof}

\begin{lemma}
    There is a malicious attack that learns all preferences in the  Propose$\&$Reject algorithm, where the malicious side initiates proposals. This attack takes $n(n-1)$ algorithm iterations with at least two communication rounds in each iteration.
\end{lemma}
\begin{proof}
    Using the round-robin proposal scheme, by Lemma~\ref{lem:first-choice-learnable}, each honest agent's top choice is revealed within the first $n$ iterations. Over the next $n(n-2)$ iterations, the attacker can uncover all remaining preferences, revealing one per round. To reveal agent $h$'s $i^{\text{th}}$ preference $i^h$, the malicious agents manipulate their rankings as follows: The agents in positions $1^h$ through $(i-1)^h$ of $h$ place $h$ at the bottom of their preference lists. All other malicious agents rank $h$ as their top choice and propose. As shown in the proof of Theorem~\ref{thm:dishonest-propose}, $h$ will accept the proposal from $i^h$. This allows the attacker to learn $h$'s $i^{\text{th}}$ choice. This process is repeated for each agent and each preference position, enabling complete recovery of all preferences within $n(n-1)$ iterations. 
\end{proof}

    

\section{Simulations}
To complement our theoretical results we additionally performed simulations on both real-world as well as synthetically generated preferences. Our main focus here will be the central Propose$\&$Reject algorithm in which the honest side initiates the proposals. From Theorem~\ref{thm:propose-reject-privacy-preserving} we know that if all first preferences of the honest agents are different, then we can not learn anything beyond the first preferences. However, by Lemma~\ref{lem:all-same-good-propose} we know if all preferences are identical, then a Round-Robin attack learns all preferences in $n$ iterations. We aim to investigate experimentally the structural gap between the two results, and answer how much we can learn in the cases in-between, specifically also in the case of structured real-world preference data.

We note that we have also simulated the case in which the dishonest agents initiate proposals to see if our algorithms learn faster in practice than what is guaranteed by Theorem~\ref{thm:dishonest-propose}, which unfortunately was not the case. In general, we learn only one preference per iteration and not significantly more information, unless in specific cases such as when all honest agents' preferences are identical.

For completeness, we describe what information we learn from a matching, as well as the strategies we utilize in more detail.

Finally, we briefly evaluate different strategies for the decentralized setting. As Lemma~\ref{lem:decentralized-learn-everything-simple} shows, by simply rejecting all incoming proposals all preferences can be learned. However,this will also result in an empty matching. We evaluate different strategies that try to learn the preferences while also computing a matching.

\subsection{Learning rules}
We initialize for every honest agent $h$ and a position $x$ in its preference with a set of possible dishonest agents that could occupy that position. Initially all these sets are $D$.
During the execution of our algorithms we remove agents from these sets if we learn any information. The goal is to reduce the size of each possibility set down to one agent, if possible. We say a preference of some honest agent at some position has been \emph{learned}, if the corresponding possibility set contains only a single element.

It may happen that at the end of our algorithms we still have more than one possibility per honest agent and position. To measure the gained information about the preferences we introduce an \emph{uncertainty} metric which we define as follows:
$$\frac{\sum_{h\in H, x \in [n]} (\text{possibilities for preference } x \text{ of agent } h) - n^2}{n^3 - n^2}$$

Consequently, the uncertainty ranges from 0, if all preferences were learned, to 1, in the case that no information on the preferences was obtained. 
We use the below learning rules to shrink the size of the possibility sets:


\subsubsection{General learning rule}
Let $\Pi^H$ and $\Pi^D$ be preferences of the honest and dishonest agents, respectively, and $M$ a stable matching under these preferences.

Let $h$ be any honest agent and let $d \in D$ be the dishonest agent it was matched to in $M$.
Further, let $P$ be the set of dishonest agents (excluding $d$) that would prefer $h$ over their current partner in $M$. From the stability property we infer that $d \succ_h p$ for all $p \in P$. This implies that $d$ cannot be in the bottom $|P|$ preferences of $h$. We thus remove $d$ from $h$'s bottom most $|P|$ possibility sets.

\subsubsection{Information propagation} Once we know there is only one possible position of agent $d \in D$ in the preferences of agent $h \in H$ then we fix that preference, by assigning the possibility set $\{d\}$ at that position. Note that more involved propagation techniques are possible, e.g. based on similar combinatorial arguments. However, for simplicity we only use this one propagation rule for its simplicity and the ability to efficiently implement it.

\subsubsection{Attack strategies}
For our experiments we utilized the following strategies for the dishonest agents:
\paragraph{Random} Here we simply try $n^2$ random preference matrices for the attackers as a baseline to compare our strategies against. We report on average values obtained from 10 different runs.
\paragraph{Round-Robin} This is the strategy utilized in Lemma~\ref{lem:all-same-good-propose}. That is, all dishonest agents initially set their preferences to $h_1 \succ \dots \succ h_n$. This preference order is then cycled in round-robin fashion in subsequent iterations. Recall that this strategy is guaranteed to always learn at least the first preferences of the honest agents.
\paragraph{Brute-Force} Here the attackers simply try all possible preferences for the attack. This takes $(n!)^n$ many iterations, which we have found to be only computationally viable for $n$ at most 4.
\paragraph{Targeted-Propose} This strategy is based on the strategy described in 
Theorem~\ref{thm:dishonest-propose}, which is guaranteed to work if the dishonest agents propose. We now consider the setting where the honest agents propose. We adapt the strategy in the following ways: we try to first learn the first preferences of every agent, then the second preference of every agent and so on. In each iteration when we are trying to learn the $i^{\text{th}}$ preference of honest agent $h$, as before, we try to prevent the dishonest agents $1^h, \dots, (i-1)^h$ from matching with $h$ by assigning $h$ as their last preference, and the remaining ones assign $h$ as their first preference. If we have not determined the preferences $1^h, \dots, (i-1)^h$ we then stop trying to learn more preferences for $h$. Unfortunately, as the honest agents propose, we can not ensure that the agents $1^h, \dots, (i-1)^h$ will not be matched to $h$. We try to mitigate this by assigning their top preferences to honest agents that also rank them highly. Specifically, we iterate over preference positions $x$ starting with 1 and ending at $n$ and for every honest agent $h' \neq h$, if we have determined the preference $x^{h'}$ we assign $h'$ at the highest available free position in $x^h$'s preferences. Recall, that agents that have each other as their first preference have to be matched in any stable matching. For completeness we give the pseudocode in Algorithm~\ref{alg:targeted-propose}.

\begin{algorithm}[t]
    Learn first preferences using Round-Robin\;
    \For{$h \in H$}{
        \For{$i \gets 1,\dots,n$}{
            \If{$1^h,\dots,(i-1)^h$ are not determined}{
                \textbf{break}\;
            }
            $F \gets \{1^h, \dots, (i-1)^h\}$\;
            \tcp{Try to learn the $i$'th preference of $h$}
            $\Pi^D[d, k] \gets \bot \quad\quad\quad\quad \forall d \in D, k \in [n]$\;
            \For{$f \in F$}{
                $\Pi^D[f] \gets [\,]$ \tcp*{Empty list}
            }
            \For{$i' \gets 1,\dots,n$}{
                \For{$h' \in H$ \textbf{with} $h' \neq h$}{
                    $f \gets (i')^{h'}$\;
                    \If{$f$ is determined and in $F$}{
                        append $h'$ to $\Pi^D[f]$
                    }
                }
            }
            \For{$f \in F$}{
                \While{$|\Pi^D[f]| < n$}{
                    append $\bot$ to $\Pi^D[f]$
                }
            }
            \For{$d \in D \setminus F$}{
                $\Pi^D[d,1] \gets h$
            }
            Replace $\bot$'s in $\Pi^D$ randomly to obtain a valid preference ordering for every agent\;
            Submit preferences and apply learning rules to the resulting matching\;
        }
    }
    \caption{Targeted-Propose}
    \label{alg:targeted-propose}
\end{algorithm}

\subsection{Brute-force attacks for $n=4$}
Firstly, we investigate how the different strategies perform by considering all possible $4\times4$ preference matrices of the honest agents. To reduce the number of possible instances we only consider non-isomorphic instances. For example, we can assume any permutation of the honest and dishonest agents without changing the problem. Therefore we only consider $4\times4$ preferences matrices for the honest agents where the entries in the first row are sorted, and the rows of the matrix are sorted lexicographically. This reduces the number of instances from 331,776 to 2,600. Note that within the Brute-Force strategy the same trick can not be performed and we still have to try all 331,776 preference matrices for the dishonest agents. Unfortunately, these were the largest instances we still could apply the Brute-Force strategy to in reasonable time.

We summarize these results in Figure~\ref{fig:bruteforce}. It can be observed that the Random strategy with $n^2$ iterations appears to perform similar to the Round-Robin strategy that only needs $n$ rounds. Even Brute-Force is not able to learn many preferences, however Targeted-Propose is not much worse despite using only at most $n^2$ iterations. Lastly, the Round-Robin approach with $n$ iterations learns significantly less than Targeted-Propose.
\begin{figure}[t!]
    \centering
    \includegraphics[width=0.4\textwidth]{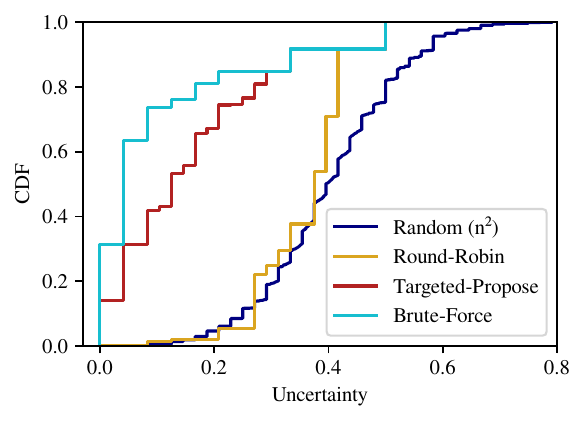}
    \caption{Results for brute-forcing all possible preferences for the honest agents using all possible preferences for the dishonest agents.}\label{fig:bruteforce}
\end{figure}

\subsection{Synthetic data - entropy dissimilarity}
We define a new measure of similarities among the preferences of the honest agents, namely the \emph{entropy dissimilarity} as follows:
$$ \frac{\sum_i \text{entropy in column } i}{n\log{n}} $$ 
This way if all preferences are identical then dissimilarity is 0, and if they are all different, then dissimilarity is 1. We sample matrices with a desired dissimilarity, by first sampling a random matrix and then doing random ranking swaps that change dissimilarity in the right direction. This way we are able to interpolate smoothly between matrices where all preferences are different and ones where they are all identical. This allows us to investigate the structural gap between Theorem~\ref{thm:propose-reject-privacy-preserving} and Lemma~\ref{lem:all-same-good-propose}. We measured the uncertainty after an attack with our learning strategies in Figure~\ref{fig:artificial}. The Random strategy with $n^2$ iterations appears to have high uncertainty for all dissimilarity values, especially for dissimilarity 1 where it has the highest uncertainty among the different strategies, meaning it was not always able to learn the first preferences. The Round-Robin strategy efficiently learns preferences only when dissimilarity is low, whereas the uncertainty of preferences after an attack with Targeted-Propose appears to correlate well with the dissimilarity. However, note that the dissimilarity is not a perfect indicator; the first preferences could be all different, and the remaining preferences mostly identical, leading to low dissimilarity, but high uncertainty.
\begin{figure}[t!]
    \centering
    \includegraphics[width=0.4\textwidth]{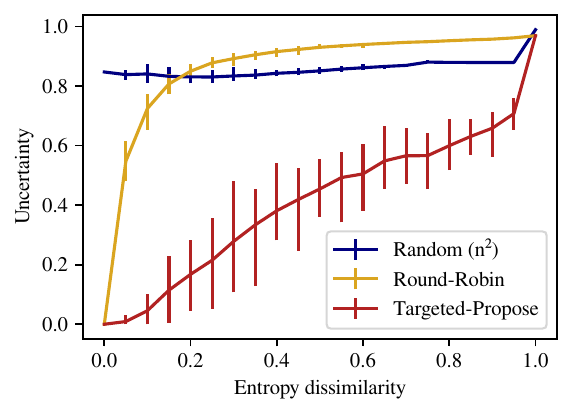}
    \caption{Uncertainty for $64\times64$ preference matrices with varying dissimilarity. The error bars show the minimum and maximum among 50 samples.}\label{fig:artificial}
\end{figure}

\subsection{Synthetic data - Mapel}
We utilized the Mapel dataset from \cite{boehmer2023mapel}. This dataset contains 503 stable matching instances grouped into 15 categories, based on the statistical cultures they originate from. All instances have the same number of agents, that is $n=100$. We take both preference matrices from each instance to obtain 1006 privacy attack instances.
See Figure~\ref{fig:mapel_data} for a visualization of the learned preferences using our Targeted-Propose approach. The following can be deduced from the data: for most instances there exists an agent whose top-75\% or more preferences we can learn fully, and in at least half of the instances we learn at least the top-25\% preferences of half of the agents. The results are therefore heavily skewed: a fraction of agents have preferences that are easier to learn compared to the majority, and a further fraction has preferences that are difficult to learn much further beyond the first preference.
\begin{figure}[t]
    \centering
    \includegraphics[width=0.4\textwidth]{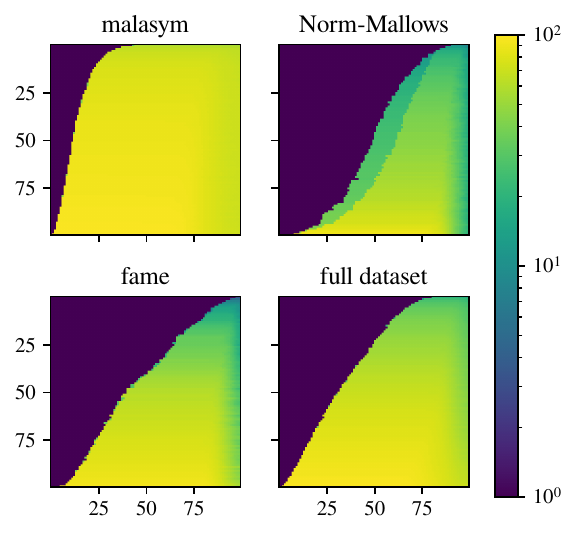}
    \caption{Learned preference information for mapel datasets. The y-axes denote the agents, and the x-axis denotes the preferences. The colors indicate the size of the possible preference sets at the positions in the preference matrices after applying the Targeted-Propose strategy. Three subgroups of data from the dataset have been selected for which we computed the element-wise median of the set sizes in the matrices, and additionally the median over all 1006 instances in the dataset was taken (bottom right). Before the aggregation with the median, the agents are sorted according to how much was learned for them.}\label{fig:mapel_data}
\end{figure}

\subsection{Real-world data - student/project allocation}
The availability of real-world matching data is unfortunately very limited. The PrefLib preference library~\cite{mattei2013preflib} contains a single real-world dataset for a matching problem with preferences, which contains data from a student/project allocation problem at a university~\cite{kwanshie2015projectallocation}. For 8 different academic years, around 31--51 students each year submitted a linear ordering over their top-5 rated projects, out of 56--155 available projects. In this case, we assume the attackers are the project supervisors, wishing to learn the project preferences of the students. To fit this to our setting we set the remaining project preferences of the students randomly. Furthermore, as the number of students is smaller than the number of available projects, we considered different strategies for ensuring the number of participants on each side is the same: duplicating one student many times, or duplicating all students roughly the same number of times, or, finally, adding new dummy students with random preferences. However, there was no large difference in the results in any of these approaches, thus we report on the case where we duplicated one student many times. Using our Targeted-Propose strategy we can determine the top 5 preferences for most students. For a visualization of the learned preferences see Figure~\ref{fig:preflib_data}. Therein we depict the four project years where we learned the fewest preferences. For the last three project years we learn all top-5 preferences, and for another 3 years we only do not learn 1--4 preferences. As we are only interested in the top-5 preferences it is also sufficient to use $5n$ attack iterations. Our Round-Robin approach was only able to learn the top-1 preferences and in a few cases also the second preference, but not much more.

\begin{figure}[t]
    \centering
    \includegraphics[width=0.4\textwidth]{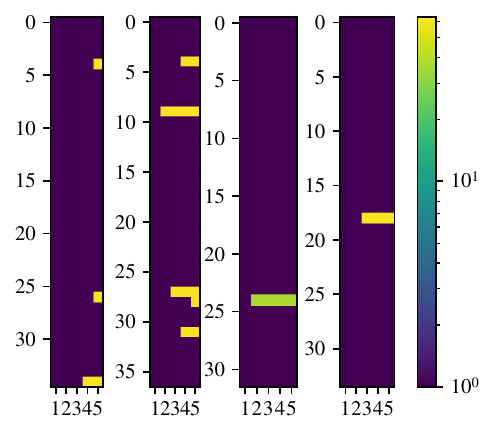}
    \caption{Learned student preferences for the academic years 2007 to 2010 (left to right) using Targeted-Propose. The students are on the y-axis and the preferences are on the x-axis, starting with top preferences on the left side. The colors indicate the size of the preference possibility sets.
    }\label{fig:preflib_data}
\end{figure}



\subsection{Decentralized Setting}
\begin{figure*}[tb]
    \centering
    \includegraphics[width=0.8\textwidth]{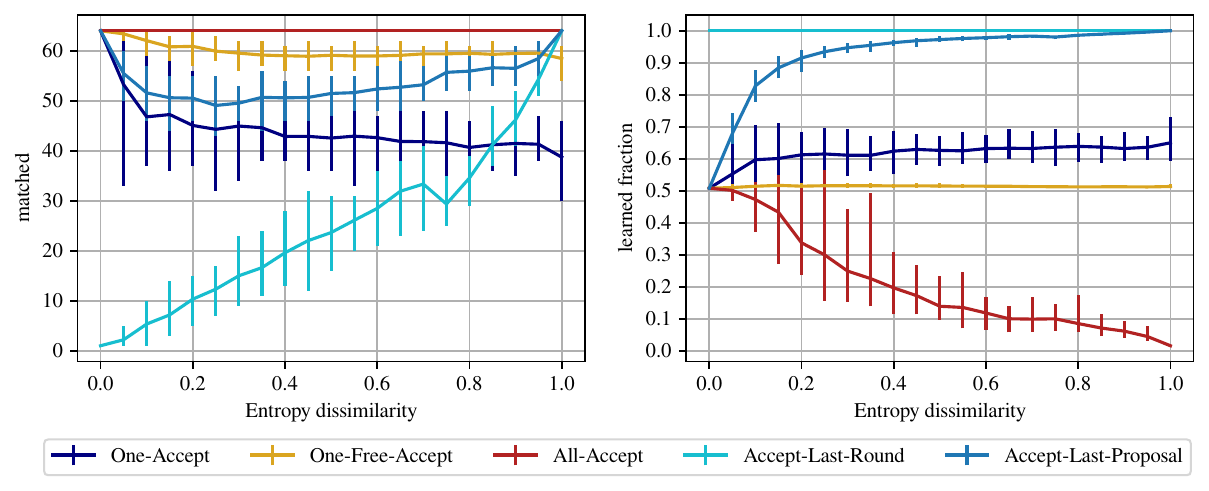}
    \caption{Simulation results for $64\times64$ preference matrices with varying dissimilarity in the decentralized setting. The error bars show the minimum and maximum among 50 samples. Matched refers to the number of matched agents after applying a strategy, whereas the learned fraction is the fraction of revealed honest agents' preferences.}\label{fig:decentralized-simulations}
\end{figure*}

We simulate a decentralized setting in which the honest agents send proposals to the dishonest agents. We assume the agents send proposals as in the Propose$\&$Reject algorithm, that is, in order of their preferences. The dishonest agents therefore learn from each incoming proposal one preference from an honest agents ranking. Furthermore, we assume the proposals are sent in synchronous rounds, where in one round all honest agents that are unmatched each send one proposal to an dishonest agent, then the dishonest agents decide which proposals to accept. The goal for the dishonest agents is therefore to behave strategically in a way that maximizes the number of proposals sent by the honest agents. As shown in Lemma~\ref{lem:decentralized-learn-everything-simple}, if the dishonest agents simply reject all incoming proposals, then they learn all preferences. However, this also results in an empty matching, which means that the honest agents can realize, that the dishonest agents did not follow the Propose\&Reject algorithm.

We empirically evaluate different strategies for the dishonest agents to accept or reject the incoming proposals in an attempt to learn as much information as possible while also matching as many agents as possible. We consider the following strategies:

\paragraph{One-Accept} in any given round an arbitrarily chosen proposal is taken and the two agents, the one being proposed to and the one proposing, are matched. If the agent that was proposed to was already matched, then it rejects its previous matched partner.
\paragraph{One-Free-Accept:} here in any round we choose arbitrarily one proposal to an unmatched dishonest agent and match it with the proposing honest agent.
\paragraph{All-Accept} all proposals are accepted by the dishonest agents. Note that if an agent gets two or more proposals in one round, then it will accept all of them, but also immediately reject all but one of them.
\paragraph{Accept-Last-Round} the dishonest agents reject all proposals in the first~$n-1$ rounds, and in the final round accept all proposals. This is similar to the strategy used in Lemma~\ref{lem:decentralized-learn-everything-simple}, but tries to find a matching in the last round.
\paragraph{Accept-Last-Proposal} here the dishonest agents keep track of which proposals were already sent by the honest agents. An unmatched dishonest agent accepts a proposal only if all other unmatched honest agents proposed to it, that is, this is the last proposal the agent may still receive, so it accepts it.

The results for these strategies are presented in Figure~\ref{fig:decentralized-simulations}. The only strategy among these that always resulted in a complete matching was All-Accept, unfortunately this was also the strategy that generally learned the least. The One-Accept and One-Free-Accept appear to match at least half of the agents, but learn only around half of the preferences. Accepting in the last round learns all preferences similar to the strategy in Lemma~\ref{lem:decentralized-learn-everything-simple}, but does the number of matched agents varies greatly --- all agents are matched only if the last preferences are all different. Accepting the last proposal appears to generally learn many preferences, however a small fraction of unmatched agents remains.

\section{Conclusion}

In this work, we introduced the notion of privacy attacks for the stable marriage problem. We showed that stable marriage algorithms from the literature are vulnerable to these attacks. Our practical evaluation demonstrated that also on real-world datasets we can often learn top-20\% of the preferences, even in the theoretically challenging settings. 

Our work shows that there is a need for new privacy-preserving stable marriage algorithms to make sure that the private preferences of the participants are protected. 
It will be interesting in future work to contribute to devising such algorithms, also answering the question: given a preference matrix of the honest participants, how much information would the best-possible matching algorithm have to reveal in this case? 
Another interesting direction for future work is to focus on attacks in which different subsets of agents participate across repeated executions. In many real-world markets, such as job markets, school choice or organ exchange, participants may enter and leave over time. Understanding how dynamic participation affects privacy leakage is an important next step. Finally, while existing privacy-preserving techniques can in principle protect preferences, their computational cost and practical feasibility for large matching markets require further investigation.

\bibliographystyle{IEEEtran}
\bibliography{our-bib-file}

\end{document}